\documentclass[journal,final,10pt]{IEEEtran}
\ifCLASSINFOpdf
\else
\fi
\usepackage{amsfonts,amsmath,amssymb}
\usepackage{mathtools}
\usepackage{mdwmath}
\usepackage{cuted}
\usepackage{mdwtab}
\usepackage{amsthm}           
\usepackage{bm}               
\usepackage{units}            
\usepackage{graphicx}
\usepackage{epstopdf}
\usepackage[linesnumbered,ruled,vlined]{algorithm2e}
\usepackage{url}

\usepackage{paralist}
\usepackage{color}
\usepackage{authblk}
\usepackage{acronym} 

\newcommand{\Fig}[1]{Fig.~\ref{fig:#1}}

\newcommand{\Eq}[1]{Eq.~(\ref{eq:#1})}

\newtheorem{proposition}{Proposition}

\newcommand{\Lc}{\mathcal{L}}

\newcommand{\Dc}{\mathcal{D}}
\newcommand{\Bc}{\mathcal{B}}
\newcommand{\Cc}{\mathcal{C}}

\newcommand{\Ac}{\mathcal{A}}

\newcommand{\Rs}{\mathbb{R}^2}

\newcommand{\Ed}{\mathbb{E}}

\newcommand{\Pd}{\mathbb{P}}
\newcommand{\sinr}{\mathrm{SINR}}

\newcommand{\dr}{\mathrm{d}}
\newcommand{\Rcone}{u(\omega,\gamma)}

\newcommand{\Upper}{u}
\newcommand{\Lower}{l}
\newcommand{\rmax}{r_\text{max}}

\newcommand{\los}{\texttt{l}}
\newcommand{\nlos}{\texttt{n}}


\newtheorem{theorem}{Theorem}
\newtheorem{corollary}{Corollary}

\DeclarePairedDelimiter\floor{\lfloor}{\rfloor}

\DeclareMathOperator*{\argmax}{\arg\!\max}

\begin{document}
%
\title{A Stochastic Model for UAV Networks Positioned Above Demand Hotspots in Urban Environments}
%
%
%

\author{Boris Galkin,
        Jacek~Kibi\l{}da,
        and~Luiz~A. DaSilva
}

\affil{CONNECT, Trinity College Dublin, Ireland \\
\textit{E-mail: \{galkinb,kibildj,dasilval\}@tcd.ie}}

\maketitle

\begin{abstract}
Wireless access points on unmanned aerial vehicles (UAVs) are being considered for mobile service provisioning in commercial networks. To be able to efficiently use these devices in cellular networks it is necessary to first have a qualitative and quantitative understanding of how their design parameters reflect on the service quality experienced by the end user. In this paper we use stochastic geometry to characterise the behaviour of a network of UAV access points that intelligently position themselves above user hotspots, and we evaluate the performance of such a network against cases where the UAVs are positioned in a rectangular grid or according to heuristic positioning algorithms.
\end{abstract}


%
\IEEEpeerreviewmaketitle

\section{Introduction}
Radio infrastructure mounted on unmanned aerial vehicles (UAVs) has become recognised by the wireless community as a core opportunity for next-generation communication networks \cite{Zeng_20162}. UAVs introduce several improvements over conventional infrastructure, such as the ability to intelligently adjust their positions in real-time and the ability to benefit from unobstructed wireless channels in the air. These advantages have made UAVs attractive candidates for a range of applications, such as wireless sensor networks \cite{Heimfarth_2014}, public safety networks \cite{Merwaday_2015}, and as flying small cells covering demand hotspots in densely populated areas \cite{BorYaliniz_20162}.

The introduction of UAVs also presents new challenges to network deployment, as the uprecedented flexibility of the new infrastructure requires more insight into how the new network variables affect the achievable performance. As the UAVs can move in three dimensions on-demand they can pursue a variety of deployment strategies with respect to the locations of end users, each other, or a combination thereof. The selection of the appropriate deployment requires an understanding of the performance impact of each deployment strategy in a given situation. In addition to this, the operating UAV height will affect the overall network performance and needs to be selected with care. Based on current regulations from the Federal Aviation Administration (FAA) \cite{FAA_2016} and the European Aviation Safety Agency (EASA) \cite{EASA_2017} as well as proposed airspace management schemes \cite{SESAR_2017} we expect that UAVs serving urban hotspots will take the form of small, lightweight (below 25kg) devices operating at heights at or below 200m. Given this height range the UAV network may be operating above a built-up urban area or below building heights in so-called urban canyons, which will significantly affect the radio environment of the UAV network. In our previous work \cite{Galkin_2017} we used stochastic geometry to provide analytical expressions for the coverage performance of a randomly distributed UAV network in such an urban environment, where buildings affect the wireless channel. In this paper we extend our work by considering the performance of the UAV network when the UAVs are positioned intelligently above hotspots of end users, thus covering the demand hotspots. Our contribution can be stated as follows:

\begin{enumerate} 
\item We provide an analytical expression for the coverage probability and average spectral efficiency for a typical user served by a network of UAVs which are positioned intelligently above the centers of user hotspots. Our model takes into account parameters such as building density, user hotspot radius and UAV antenna beamwidth, and can represent different wireless fast-fading behaviours through generalised Nakagami-m fading. 
\item  Using our model we demonstrate that there exists an optimum UAV height for a given user hotspot radius, and that larger user hotspots require UAVs to increase their altitudes to maximise performance. We also demonstrate how the optimum UAV height is almost unaffected by varying the density of user hotspots and UAVs. Our numerical results demonstrate how the presence of interference in the UAV network imposes a strict limitation on the range of heights the UAV network can operate at.
\item We compare the performance of the UAV network deployed above user hotspots against UAV networks which are deployed randomly with respect to users, UAV networks which are deployed according to a heuristic optimisation algorithm, as well as UAV networks deployed in a rectangular grid. This comparison allows us to demonstrate that for higher UAV altitudes and larger UAV densities the UAV network benefits more from UAVs positioning themselves with respect to each other rather than with respect to user hotspots, due to the effect of mitigating interference. This leads us to conclude that there exists a certain UAV threshold density above which the UAV network should not position itself around user hotposts but should instead spread out the UAVs as much as possible to reduce coverage overlap and interference.
\end{enumerate}

\section{Related Work}
The wireless community has published a number of works on the topic of deploying UAV-mounted access points to serve terrestrial users. These works typically set up an optimisation problem where a UAV parameter such as location in 3D space is optimised, subject to a given objective function. In \cite{Zeng_2016} the authors consider a single fixed-wing UAV acting as a relay between two ground terminals, and the throughput is maximised by optimising the trajectory and the transmit power of the UAV across discrete timeslots. The authors demonstrate how the trajectory can be optimised using convex optimisation methods when the transmit powers are fixed, and propose an iterative algorithm for the case where both the trajectory and the transmit powers are variable. The same authors optimise the trajectory to maximise energy-efficiency in \cite{Zeng_2017}, and they optimise transmit power for a UAV with a circular orbiting pattern in \cite{Zhang_20172}. In \cite{BorYaliniz_2016} the authors consider the problem of placing a UAV in 3D space to maximise the number of users that are covered, subject to a quality of service (QoS) constraint. The authors describe the scenario as a mixed integer non-linear problem which they solve given different environmental parameters. In \cite{Fotouhi_2016} the authors consider a UAV which dynamically repositions itself with respect to randomly moving users to maximise the achieveable spectral efficiency at every discrete timeslot. The authors quantify the performance gain with respect to a fixed UAV case, given different bandwidth allocation policies. In \cite{Chen_2017} the authors consider a scenario where a UAV relay iteratively searches an urban environment for locations where it can establish a LOS connection to the end user and BS, meeting the channel rate requirements. The authors present a converging algorithm which significantly outperforms direct BS-user communication. In \cite{Ouyang_2018} the authors optimise the trajectory of a laser-powered UAV access point to maximise achieveable throughput while meeting energy constraints. The authors propose a flight pattern which involves the UAV orbiting around the laser transmitter to harvest energy before hovering above the end user to maximise throughput. In \cite{Gruber_2016} the author considers a scenario where each UAV in a network of UAVs iteratively adjust its location to maximise its spectral efficiency, subject to interference from the remaining UAVs. The author demonstrates that the UAVs are able to converge on locally optimum 2D coordinates but are not able to converge on an optimum height above ground. In our previous work \cite{Galkin_2016} we demonstrated how a generic K-means clustering algorithm can be used to position UAVs around user locations in an interference-free environment, and that the resulting UAV network outperforms fixed terrestrial networks in terms of received signal strength at the user locations. K-means clustering was also proposed for UAV placement optimisation in \cite{Mozaffari_20162}.

We report an emerging trend which can be observed in the works cited above. With the notable exception of \cite{Gruber_2016}, existing state of the art on UAV network optimisation tends to ignore the effects of interference, instead focusing on scenarios where individual UAVs are operating in isolation. Without interference the wireless links are limited by the geometry of the environment, and therefore the network performance is generally optimised through minimising the distance between the UAV and the receiver, as this minimises the pathloss, increases the LOS probability and enables the network to reduce transmit power. These optimisation strategies may not apply to a scenario where multiple UAVs are operating concurrently and creating interference for each other. In the presence of interference, decreasing the distances between transmitters and receivers may also have the result of decreasing the distances between interferers and receivers, potentially causing a net decrease in channel performance. 

 Stochastic geometry is an alternative method for modelling the spatial relationships in a UAV network, capturing the effect of interference on network performance and giving us new insight into the performance trade-offs of UAV networks. Stochastic geometry has seen widespread use for analysing the behaviour of terrestrial networks; notable examples include \cite{AndrewsBaccelliGanti_2011} in which the authors derive expressions for the coverage probability and mean rate given uniformly and randomly distributed base stations, and \cite{Saha_2017} in which the authors consider multi-tier network deployments with small cells placed around user hotspots. UAV networks have also been analysed using stochastic geometry in prior art. In \cite{Ravi_2016} and \cite{Chetlur_2017} the authors derive the coverage probability for a stochastic UAV network under guaranteed LOS conditions for a fading-free and Nakagami-m fading channel. The authors describe a fixed number of UAVs operating within a fixed area at a certain height above ground and demonstrate how an increase in height results in a decrease in the coverage probability. Additionally in \cite{Chetlur_2017} they demonstrate how larger values of fading parameter $m$ reduce the variance of the random signal-to-interference ratio (SIR) experienced by the user. Stochastic geometry is applied by the authors of \cite{Zhang_2017} to optimise UAV density in a radio spectrum sharing scenario under guaranteed LOS conditions. In \cite{Hayajneh_20162} the authors evaluate the performance of a network of UAVs acting alongside a terrestrial BS network in an emergency outage scenario. The authors of \cite{MahdiAzari_20172} and \cite{MahdiAzari_20173} use stochastic geometry to evaluate the performance of a terrestrial BS network that is serving terrestrial users and UAV users simultaneously. In our previous works \cite{Galkin_2017} and \cite{Galkin_2018} we model the coverage probability of the user access and wireless backhaul links, respectively, of a UAV network operating in an urban environment in the presence of LOS-blocking buildings, assuming independent distributions of the transmitters and receivers. This paper extends our work in \cite{Galkin_2017} by extending our stochastic geometry model to account for UAVs that are positioned above known user hotspot locations. This new model allows us to investigate the effect that intelligent UAV placement has on network performance, compared against UAV placement that is blind to user locations or the locations of other, interfering UAVs.

\section{System Model}

\begin{figure}[t!]
\centering
	\includegraphics[width=.40\textwidth]{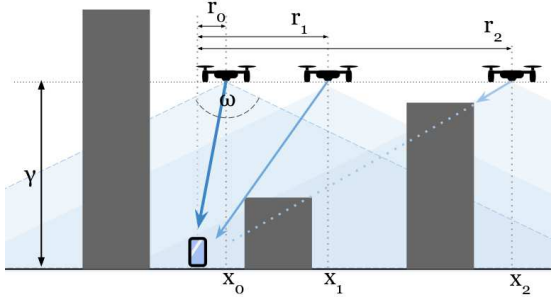}
	\caption{
	Side view showing UAVs in an urban environment at a height $\gamma$, with 2D coordinates $x_0,x_1$, $x_2$ and antenna beamwidth $\omega$. The user is serviced by the UAV with the strongest signal, while the remaining UAVs generate LOS and NLOS interference.
	\vspace{-5mm}
	}
	\label{fig:drone_network}
\end{figure}

We consider a scenario where a number of users congregate in an area of interest, creating several user clusters which generate data demand. These clusters are referred to as demand hotspots, which the UAV network attempts to serve. The set of hotspots in the area of interest is modelled as a Matern Cluster Process (MCP) \cite{Saha_2017}. The number of user hotspots in the area of interest is random, with average hotspot density $\lambda_p$. The location of each hotspot is random in $\Rs$ and is independent of the location of other hotspots; the set of hotspot centers is denoted as $\Phi_p  = \{y_0,y_1,...\} \subset \Rs$ where $y_i$ corresponds to the $i$th hotspot center. The users belonging to a hotspot $i$ are positioned in a circle of radius $\rmax$ centered on the hotspot center $y_i$. Users are randomly and uniformly positioned inside this circle.

We perform our analysis for a reference user, which is selected as the random user of a random hotspot. As the set of hotspots is stationary and translation-invariant we assume the reference user to be located in the origin, and we denote $y_0 \in \Phi_p$ as the location of the reference user hotspot center. 

There exists a network of UAVs operating at a fixed height $\gamma$ above ground. We denote the set of UAVs as $\Phi_u = \{x_0,x_1,...\} \subset \Rs$ where $x_i$ corresponds to the projected coordinates of the $i$th UAV onto $\Rs$. The reference user will be served by one of the UAVs in $\Phi_u$; the remaining UAVs will generate interference for the reference user.

The serving UAV is the one from which the reference user observes the highest received power. We denote its location as 
\begin{equation}
x_\star = \argmax_{x_i\in\Phi_u}\{\bar{S}_{i}\}
\label{eq:Ass}
\end{equation}

\noindent
where $\bar{S}_{i}$ denotes the long-term average\footnote{Since cell-level association acts on the order of seconds, we assume that any fast fading effects (like the multipath fading) will be averaged out.} power received from the $i$th UAV at $x_i$.

\subsection{Propagation Model}

We assume that the propagating signal is subject to distance-dependent pathloss and small-scale power fading, following the channel models used in \cite{Saha_2017}, \cite{Chetlur_2017} and \cite{MahdiAzari_20172}. The instantaneous received signal strength from a UAV $i$ at $x_i$ at the reference user is:
\begin{equation}
S_i = \eta_{r_i} H_{T_i} l(r_i,\gamma,\alpha_{T_i})
\end{equation}

\noindent
where $r_i = ||x_i||$ is the horizontal distance to the UAV, $\eta_{r_i}$ is the antenna gain in the direction of the reference user, $H_{T_i}$ is the random multipath fading component experienced by the UAV $i$, $l(r_i,\gamma,\alpha_{T_i})=(r_i^2+\gamma^2)^{-\alpha_{T_i}/2}$ is the pathloss function, $\alpha_{T_i}$ is the pathloss exponent, $T_i \in\{\los,\nlos\}$ is a random variable denoting whether UAV $i$ has a LOS or NLOS channel type to the user. Note that the channel type will determine the value of the pathloss exponent $\alpha_{T_i}$ as well as the small scale fading $H_{T_i}$. We assume Nakagami-m small scale fading, as this is a generalised fading model which can represent a variety of radio environments \cite{Chetlur_2017}. Under Nakagami-m fading the component $H_{T_i}$ is a gamma-distributed random variable, with fading parameter $m_{T_i}$.

\subsubsection{Transmit Power and Antenna Gain} We assume UAVs have identical transmit power $\mu$ and a directional antenna with beamwidth $\omega$. The main beam illuminates the area directly beneath the UAV. We assume a uniform and rotationally symmetric beam pattern; using the approximations (2-26) and (2-49) in \cite{Balanis_2005} and assuming perfect antenna radiation efficiency the gain  $\eta_{r_i}$ in the direction of the reference user from UAV $i$ can be expressed as
\begin{equation}
\eta_{r_i} = \begin{cases} \mu 16\pi/(\omega^2), \quad \text{if}\quad r_i\leq \Rcone \\
0, \quad \text{if}\quad r_i>\Rcone
\label{eq:antgain}
\end{cases}
\end{equation}

\noindent
where $\Rcone=\tan(\omega/2)\gamma$.

\subsubsection{Channel Model}
 In an urban environment the LOS is blocked by buildings between the transmitter and receiver. To model a distribution of buildings we adopt the model in \cite{ITUR_2012}, which defines an urban environment as a collection of buildings arranged in a square grid. There are $\beta$ buildings per square kilometer, the fraction of area occupied by buildings to the total area is $\delta$, and each building has a height which is a Rayleigh-distributed random variable with scale parameter $\kappa$. The probability of a UAV $i$ having LOS channel to the reference user $(T_i = \los)$ is given in \cite{ITUR_2012} as 
\vspace{-1mm}
\begin{equation}
\Pd_{\los}(r_i) = \prod\limits_{n=0}^{\max(0,d-1)}\left(1-\exp\left(-\frac{\left(\gamma - \frac{(n+1/2)\gamma}{d}\right)^2}{2\kappa^2}\right)\right)
\label{eq:LOS}
\vspace{-1mm}
\end{equation}
\noindent
where $d = \floor*{r_i\sqrt{\beta\delta}}$. It follows that the NLOS probability $\Pd_{\nlos}(r_i) = 1-\Pd_{\los}(r_i)$. 

\subsection{SINR Model}

The SINR for the reference user can be described as:
\begin{equation}
\sinr = S_{\star}/(I+\sigma^2),
\end{equation} 
where $S_{\star} = \eta_{r_\star} H_{T_\star} l(r_\star,\gamma,\alpha_{T_\star})$ is the signal from the serving UAV, $I$ denotes the aggregate signal power received from all UAVs in $\Phi_u$ other than the serving UAV and $\sigma^2$ denotes the noise power. We assume that the UAV network uses a frequency band which is orthogonal to that used by the terrestrial networks, and therefore the user does not experience interference from any terrestrial devices in the downlink.

\subsection{Coverage Probability and Spectral Efficiency}
The reference user is said to be successfully served by the UAV network if it establishes a downlink channel with a SINR above some minimum threshold $\theta$. We refer to the probability of the SINR exceeding this threshold as the coverage probability 
\begin{equation}
\Pd_c = \Pd(\sinr > \theta).
\end{equation}

Using the Shannon capacity bound the spectral efficiency (SE) of the UAV network can be expressed in terms of the SINR as 

\begin{equation}
\text{SE} = \Ed[\log_2(1+\sinr)].
\end{equation}

\section{Mathematical Analysis}

In this section we provide an analytical expression for the coverage probability of the UAV network for the case when the UAVs are deployed above user hotspots. In the first subsection we provide the basic expressions for the set of UAVs and their distances to the user, in the second subsection we derive the association probabilities of the reference user, in the third subsection we derive the Laplace transforms of the aggregate interference and noise experienced by the reference user and in the final subsection we bring our derivations together to provide our main result.

\subsection{UAV Placement and Distance Distribution}

We assume that the UAV network knows the exact number and location of user hotspots in the area of interest. The UAV network serves the hotspots by positioning exactly one UAV above the center of each hotspot. As a result, both the density and coordinates of the UAVs match those of the hotspots exactly, with $\Phi_u \equiv \Phi_p$ and $\lambda_u \equiv \lambda_p$. It follows that the reference user at the origin will have a UAV above its associated hotspot center; we denote this UAV with the index $0$ and its location as $x_0\in \Phi_u$. We partition the set $\Phi_u$ into the sets $\{x_0\}$ and $\Phi^! = \Phi_u\setminus \{x_0\}$, containing the reference user hotspot UAV and all the remaining UAVs, respectively. Note that, following Slivnyak's theorem \cite{Haenggi_2013}[Theorem 8.10], the set of UAVs $\Phi^!$ remains a PPP with intensity $\lambda_u$.

The horizontal distance between the reference user and the UAV at its hotspot center is denoted as random variable $R_0$, given an MCP distribution of user hotspots the probability density function (pdf) of $R_0$ is given in \cite{Saha_2017} as

\begin{equation}
f_{R_0}(r) = 
\begin{cases}
2r/\rmax^2 \quad \text{if} \quad 0\leq r \leq \rmax \\
0 \quad \text{otherwise}
\label{eq:hotspotpdf}
\end{cases}
\end{equation}

The distance to the hotspot center UAV at $x_0$ is independent of the channel type the UAV has to the reference user, therefore the joint pdf of $R_0$ and $T_0$ can be written as 

\begin{equation}
 f_{R_{0j}}(r) = \Pd_{j}(r)f_{R_0}(r), \quad j \in \{\los,\nlos\}
\end{equation}

The user may be served by one of the UAVs in $\Phi^!$ if it provides the strongest signal. The UAVs in the set $\Phi^!$ have a mixture of LOS and NLOS channels to the reference user, and it is more convenient for subsequent derivations to consider the behaviour of LOS and NLOS UAVs in $\Phi^!$ separately. We partition the UAVs in $\Phi^!$ into two disjoint sets $\Phi_{\los} = \{x_i \in \Phi^! : T_i = \los\}$ and $\Phi_{\nlos} = \{x_i \in \Phi^! : T_i = \nlos\}$. In effect, we carry out a thinning procedure on the PPP $\Phi^!$ \cite{Haenggi_2013}[Theorem 2.36] using $\Pd_{\los}(r)$ as a thinning function, with $\Phi_{\los}$ and $\Phi_{\nlos}$ being the resulting thinned processes. These two sets are PPP with intensity functions $\lambda_{\los}(x) = \Pd_{\los}(||x||)\lambda_u$ and $\lambda_{\nlos}(x) = \Pd_{\nlos}(||x||)\lambda_u$, respectively. The pdf of the distance to the closest UAV in $\Phi_{\los}$ and $\Phi_{\nlos}$ is defined in \cite{Haenggi_2013} as 

\begin{equation}
f_{R_{j}}(r) = 2\pi \lambda_j(r) r \exp\Big(-2\pi\int\limits_{0}^{r}\lambda_j(r) r\dr r\Big) \quad j \in \{\los,\nlos\}
\end{equation}
 
\noindent
where $R_j$ denotes the distance to the closest UAV in $\Phi_{j}$. If the user associates with a UAV in $\Phi_j$ it will associate to the closest UAV in the set, as all of the remaining UAVs in the set will, by definition, provide a weaker signal to the user. Let $V=\{\texttt{0\los},\texttt{0\nlos},\texttt{l},\texttt{n}\}$ be the random variable denoting whether the serving UAV is in $\{x_0\}$ with a LOS channel, in $\{x_0\}$ with an NLOS channel, in $\Phi_{\los}$ or in $\Phi_{\nlos}$, respectively. The pdf of the serving UAV distance $R_\star$ will follow one of the four distance distributions given above, depending on which UAV type the user associates with. The user will associate to a UAV of type $v$ if there are no other UAVs that provide a stronger signal to it. This is referred to as the association probability and is given below.

\subsection{Association Probability}

\begin{proposition}
The probability that the user's serving UAV a distance $r_\star$ away with channel type $t_\star$ is its hotspot center UAV is given as 

\begin{align}
&\Ac_{\texttt{0} t_\star} = \Bigg(\prod_{j\in\{\los,\nlos\}}\exp\Big(-2\pi\hspace{-3mm}\int\limits_{0}^{c_j(t_\star)}\hspace{-3mm}\lambda_j(r) r \dr r \Big)\Bigg),
\label{eq:Ac0}
\end{align}
\end{proposition}

\begin{proof}
The hotspot center UAV at $x_0$ will have the strongest signal if the closest UAVs in $\Phi_{\los}$ and $\Phi_{\nlos}$ are not close enough to provide a stronger signal. The probability $\Ac_{\texttt{0} t_\star}$ is then given as

\begin{align}
&\Ac_{\texttt{0} t_\star} = \nonumber\\
&\Pd\Big((r_0^2+\gamma^2)^{-\alpha_{t_\star}/2} > \max\big((R_\los^2+\gamma^2)^{-\alpha_\los/2},(R_\nlos^2+\gamma^2)^{-\alpha_\nlos/2}\big)\Big)\nonumber\\
&\overset{(a)}{=} \hspace{-3mm}\bigcap_{j\in\{\los,\nlos\}}\hspace{-2mm}\Pd\Big(R_j > c_j(t_\star)\Big)
\overset{(b)}{=} \Bigg(\prod_{j\in\{\los,\nlos\}}\exp\Big(-2\pi\hspace{-3mm}\int\limits_{0}^{c_j(t_\star)}\hspace{-3mm}\lambda_j(r) r \dr r \Big)\Bigg)
\label{eq:Ac0proof}
\end{align}

\noindent
where $(a)$ comes from the fact that $R_{\los}$ and $R_{\nlos}$ are distributed independently of each other and $(b)$ comes from the definition of the void probability of a PPP \cite{Haenggi_2013}. Here, $c_{j}(t_\star)$ denotes the lower bounds on LOS and NLOS interferer distances. These lower bounds are functions of the type of channel of the serving UAV and its horizontal distance. If the serving UAV is within LOS then $c_{\los}(\los) = r_\star$ and $c_{\nlos}(\los) = \sqrt{\max(0,(r_{\star}^{2}+\gamma^2)^{\alpha_{\los}/\alpha_{\nlos}} - \gamma^2)}$, and if the serving UAV is NLOS then $c_{\los}(\nlos) =  \min(\Rcone,\sqrt{(r_{\star}^{2}+\gamma^2)^{\alpha_{\nlos}/\alpha_{\los}} - \gamma^2})$ and $c_{\nlos}(\nlos) = r_\star$.
\end{proof}

\begin{proposition}
The probability that the serving UAV belongs to the set $\Phi_{j}$ is denoted as

\begin{align}
&\Ac_{j} =\exp\Big(-2\pi\int\limits_{0}^{c_{j\prime}(t_\star)}\lambda_{j\prime}(r) r\dr r\Big) \Bc_{\texttt{0}}, \quad j \in \{\los,\nlos\}
\label{eq:Ac1}
\end{align}

\noindent
where $j\prime$ denotes the opposite channel type of $j$ and $\Bc_{0}$ denotes the probability that the hotspot center UAV at $x_0$ is providing a weaker signal than the serving UAV from $\Phi_{j}$, and is defined as

\begin{align}
\Bc_{0} = 1-\left(\sum\limits_{k\in \{\los,\nlos\}}\int\limits_{0}^{c_k(t_\star)}\hspace{-2mm}f_{R_0k}(r)\dr r\right)\nonumber\\
\label{eq:Bc}
\end{align}

\end{proposition}

\begin{proof}
$\Ac_{j}$ is the probability that there are no UAVs in the set $\Phi_{j^\prime}\cup\{x_0\}$ which are close enough to the user to provide a stronger signal than the UAV a distance $r_\star$ away with channel type $t_\star = j$.
\begin{align}
&\Ac_{j} = 
\Pd\Big(R_{j^\prime} > c_{j^\prime}(t_\star)\Big)\Pd\Big(\bar{S}_0 <\bar{S}_\star\Big)\nonumber\\
&\overset{(a)}{=}\Pd\Big(R_{j^\prime} > c_{j^\prime}(t_\star)\Big)\nonumber\\
&\cdot\Big(1-\Big(\Pd(R_0\leq c_{\los}(t_\star),T_{0}=\los)+\Pd(R_0\leq c_{\nlos}(t_\star),T_{0}=\nlos)\Big)\Big) \nonumber\\
&=\exp\Big(-2\pi\hspace{-3mm}\int\limits_{0}^{c_{j^\prime}(t_\star)}\hspace{-3mm}\lambda_{j^\prime}(r) r \dr r \Big) 
\Bigg(1-\left(\sum\limits_{k\in \{\los,\nlos\}}\int\limits_{0}^{c_k(t_\star)}\hspace{-2mm}f_{R_0k}(r)\dr r\right)\Bigg)\nonumber\\
\label{eq:Ac1proof}
\end{align}

\noindent
where $(a)$ comes from finding the probability of the hotspot center UAV $x_0$ providing a weaker signal by taking the complement of the probability that $x_0$ is either LOS and providing a stronger signal or NLOS and providing a stronger signal.
\end{proof}

\subsection{Laplace Transform of Aggregate Interference and Noise}
The SINR is affected by the aggregate interference $I$ as well as the noise $\sigma^2$, and as $I$ is a random variable the sum of the two components is also a random variable. In this section we provide an expression for the k-th derivative of the Laplace transform of the aggregate interference and noise $\Lc_{(I+\sigma^2)}$; this will be used for derivations in the next sub-section.

\begin{proposition}
The k-th derivative of the Laplace transform $\Lc_{(I+\sigma^2)}$ is obtained as 

\begin{align}
&\frac{d^{k}\Lc_{(I+\sigma^2)}(s)}{d s^k} = \sum_{i_0+i_{\los}+i_{\nlos}+i_{\sigma}=k}\hspace{-5mm}\frac{k!}{i_0!i_{\los}!i_{\nlos}!i_{\sigma}!}\nonumber \\
&\cdot\frac{d^{i_0}\Lc_{I_0}(s)}{ds^{i_0}} \frac{d^{i_{\los}} \Lc_{I_{\los}}(s)}{ds^{i_{\los}}}
 \frac{d^{i_{\nlos}}\Lc_{I_{\nlos}}(s)}{ds^{i_{\nlos}}}\frac{d^{i_{\sigma}}\exp(-s\sigma^2)}{ds^{i_{\sigma}}} \nonumber \nonumber\\.
\label{eq:condProb3v1}
\end{align}

\noindent
where $\Lc_{I_{\los}}$, $\Lc_{I_{\nlos}}$ and $\Lc_{I_{0}}$ are the Laplace transforms of the aggregate interference from $\Phi_{\los}$, $\Phi_{\nlos}$ and $x_0$, respectively.
\end{proposition}
\begin{proof}
The aggregate interference power $I$ is the sum of the interference power $I_{\los}$, $I_{\nlos}$ and $I_{0}$ from the UAVs in $\Phi_{\los}$, $\Phi_{\nlos}$ and $x_0$, respectively. Following the PPP distribution of the hotspot centers and the independent channel type assignment these UAV sets are independent of one another; it follows that the interference random variables are independent also. This means that the Laplace transform of the aggregate interference and noise $\Lc_{(I+\sigma^2)}$ can be represented as the product of Laplace transforms $\Lc_{I_{\los}}$, $\Lc_{I_{\nlos}}$ and $\Lc_{I_{0}}$, as well as $\exp(-s\sigma^2)$. The derivative of $\Lc_{(I+\sigma^2)}$ can be expressed in the form given in \Eq{condProb3v1} following the general Leibniz rule. 
\end{proof}

 \textit{Remark 1: }
 If the user is served by the hotspot center UAV at $x_0$ then $I_0$ will be 0 and $\Lc_{I_0}$ will be 1, as the UAV will not create interference for itself. 
 
 The Laplace transforms $\Lc_{I_{\los}}$ and $\Lc_{I_{\nlos}}$ have been previously derived by us in \cite{Galkin_2017}, and the Laplace transform $\Lc_{I_0}$ for the case when the hotspot center interference $I_0$ is not zero is given below.

\begin{proposition}
The Laplace transform of the interference $I_0$, when the user is not served by the hotspot UAV, is given as 

\begin{align}
\Lc_{I_0}(s) = \frac{1}{(\rmax^2)\Bc_{0}}\sum\limits_{j\in\{\los,\nlos\}}(\Cc_j(s) + \Dc_j)
\label{eq:laplaceFinal}
\end{align}

\noindent
where 

\begin{align}
&\Cc_j(s) = \sum\limits_{q=\floor*{c_j(t_\star)\sqrt{\beta\delta}}}^{\floor*{\min(\rmax,\Rcone)\sqrt{\beta\delta}}} \hspace{-5mm}\Pd_{j}(l)\Bigg((\Upper^2-\Lower^2)+\sum\limits_{k=1}^{m_{j}}\binom{m_{j}}{k}(-1)^{k}\nonumber \\
&\cdot\bigg((\Upper^2+\gamma^2)\mbox{$_2$F$_1$}\Big(k,\frac{2}{\alpha_{j}};1+\frac{2}{\alpha_{j}};  -\frac{m_{j}(\Upper^2+\gamma^2)^{\alpha_{j}/2}}{\eta_{r_0}s}\Big) \nonumber\\
&-(\Lower^2+\gamma^2)\mbox{$_2$F$_1$}\Big(k,\frac{2}{\alpha_{j}};1+\frac{2}{\alpha_{j}};-\frac{m_{j}(\Lower^2+\gamma^2)^{\alpha_{j}/2}}{\eta_{r_0}s}\Big)\bigg)\Bigg)
\label{eq:laplaceFinalC}
\end{align}

For the case when $c_{j}(t_\star)< \min(\rmax,\Rcone)$, with $\mbox{$_2$F$_1$}(a,b;c;d)$ denoting the Gauss hypergeometric function, $l = \max(c_j(t_\star),q/\sqrt{\beta\delta})$ and $u = \min(\rmax,\Rcone,(q+1)/\sqrt{\beta\delta})$. If $c_{j}(t_\star)\geq\min(\rmax,\Rcone)$ then $\Cc_j(s) = 0$. 

$\Dc_j$ is given as

\begin{align}
\Dc_j = 
\hspace{10mm}2\hspace{-5mm}\int\limits_{\max(c_{j}(t_\star),\Rcone)}^{\rmax}\hspace{-5mm}\Pd_{j}(r_0)r_0 \dr r_0
\end{align}

\noindent
if $\rmax > \max(c_{j}(t_\star),\Rcone)$ and 0 otherwise. Note that $\Dc_j$ is not a function of $s$.

\end{proposition}

\begin{proof}
The proof is given in Appendix A. 
\end{proof}

 \textit{Remark 2: }The Laplace transform $\Lc_{(I+\sigma^2)}(s)$ requires higher-order derivatives of the Laplace transform $\Lc_{I_0}(s)$; the required analytical expressions are given in Appendix B.
 
 \subsection{General model}
 In this subsection we present the main analytical result of our paper.
 
 \begin{theorem}
 The coverage probability of the reference user served by a UAV network that positions itself above user hotspots is given as 
 
 \begin{align}
 & \Pd_c(\theta,\gamma,\lambda_u,\omega,\rmax)= \sum_{v\in\{\texttt{0l},\texttt{0n},\texttt{l},\texttt{n}\}} \nonumber\\
 &\cdot \int_{0}^{\Rcone} \Ac_{v}\sum\limits_{k=0}^{m_{t_\star}-1}(-1)^k \frac{s_r^k}{k!}\frac{d^{k}\Lc_{(I+\sigma^2)}(s_r)}{d s_r^k} f_{R_v}(r_\star)\dr r_\star
 \label{eq:general}
 \end{align}
 
 \end{theorem}
 
 \noindent
 \begin{proof}
 The coverage probability is derived as 
 \begin{align}
 &\Pd\Big(\sinr\geq\theta \Big) \nonumber\\
 &= \Pd\Big(\frac{S_{\star}}{I+\sigma^2}\geq\theta \Big)\nonumber\\
 &=\Pd\Big(\frac{\eta_{R_\star} l(R_\star,\gamma,\alpha_{T_\star})H_{T_\star}}{I+\sigma^2}\geq\theta \Big)\nonumber\\
 &=\Pd\Big(H_{T_\star}\geq\frac{\theta(I+\sigma^2)}{\eta_{R_\star} l(R_\star,\gamma,\alpha_{T_\star})} \Big)\nonumber\\
 &\overset{(a)}{=}\Ed_{R_\star,T_\star}\Big[ \Ed_{I}\Big[ \frac{\Gamma(m_{t_\star},s_r (I+\sigma^2))}{\Gamma(m_{t_\star})}\Big] \Big]\nonumber\\
 &\overset{(b)}{=}\Ed_{R_\star,T_\star}\Big[ \Ed_{I}\Big[ \exp(-s_r (I+\sigma^2))\sum\limits_{k=0}^{m_{t_\star}-1}\frac{(s_r (I+\sigma^2))^k}{k!}\Big] \Big]\nonumber\\
 &\overset{(c)}{=}\Ed_{R_\star,T_\star}\Big[ \sum\limits_{k=0}^{m_{t_\star}-1}(-1)^k \frac{s_r^k}{k!}\Ed_{I}\Big[\frac{d^{k}\exp(-s_r (I+\sigma^2))}{d s_r^k}\Big] \Big]\nonumber\\
 &\overset{(d)}{=}\Ed_{R_\star,T_\star}\Big[ \sum\limits_{k=0}^{m_{t_\star}-1}(-1)^k \frac{s_r^k}{k!}\frac{d^{k}\Lc_{(I+\sigma^2)}(s_r)}{d s_r^k}\Big]\nonumber\\
 \label{eq:generalProof}
 \end{align}
 \noindent
 where (a) comes from the random fading $H_{T_\star}$ being gamma distributed with channel-dependent fading parameter $m_{T_\star}$ with $\Gamma(.)$ and $\Gamma(.,.)$ being the gamma and upper incomplete gamma functions, respectively, where $s_r = m_{t_\star}\theta/(\eta_{r_\star} l(r_\star,\gamma,\alpha_{t_\star}))$, (b) comes from expressing the incomplete gamma function as in \cite{Ryzhik_2007}[8.352.2], (c) arises from the substitution $\exp(-s_r (I+\sigma^2))((I+\sigma^2))^k = (-1)^k d^{k}\exp(-s_r (I+\sigma^2))/d s_r^k$, (d) comes from the Leibniz integral rule. The final expression in \Eq{general} comes from taking the expectation with respect to $R_\star$ and $T_\star$, using the probability density functions given in Section IV.A and the association probabilities given in Section IV.B.
 
 Note that, following the antenna gain definition in \Eq{antgain}, a UAV will have a gain of 0 at the reference user if it is further away than $\Rcone$; as such we consider $\Rcone$ to be the upper limit on the serving UAV distance in the integral above.
 \end{proof}
 
 For comparison against a network of UAVs positioned according to a PPP at a fixed height, we present the following corollary.
 
 \begin{corollary}
 The coverage probability of the reference user when served by a UAV network that is positioned according to a PPP independently of hotspot locations is given as
 \begin{align}
  & \Pd_c(\theta,\gamma,\lambda_u,\omega)=\sum_{v\in\{\los,\nlos\}} \nonumber\\
  &\cdot \int_{0}^{\Rcone} \Ac_{v}\sum\limits_{k=0}^{m_{t_\star}-1}(-1)^k \frac{s_r^k}{k!}\frac{d^{k}\Lc_{(I+\sigma^2)}(s_r)}{d s_r^k} f_{R_v}(r_\star)\dr r_\star
  \label{eq:generalPPP}
  \end{align}
 \end{corollary}
 
 \begin{proof}
 If the UAV network is distributed independently of the hotspot centers then the reference user does not have a UAV above its hotspot center, and can only be served by a UAV from the set $\Phi^!$. The expression in \Eq{generalPPP} is obtained by setting $\rmax \to \infty$, which has the effect of setting $f_{R_0}(r_\star) \to 0$ for the range $0 \leq r_\star \leq \Rcone$, $\Bc_{0}$ and $\Lc_{I_0}\to 1$ and which reduces the expression given in \Eq{general} to the one in \Eq{generalPPP}. Note that \Eq{generalPPP} corresponds to the result presented by us in our previous work \cite{Galkin_2017}.
 \end{proof}
 
\section{Numerical Results}

\subsection{Model Performance Evaluation}

\begin{table}[b!]
\vspace{-3mm}
\begin{center}
\caption{Numerical Result Parameters}
\begin{tabular}{ |c|c| } 
 \hline
 Parameter & Value  \\ 
 \hline
 $\omega$ & \unit[150]{deg} \\
 $\alpha_{\los}$ & 2.1 \\
 $\alpha_{\nlos}$ & 4 \\
 $m_{\los}$ & 3 \\
 $m_{\nlos}$ & 1 \\
 $\mu$ & \unit[0.1]{W} \\
 $\sigma^2$ & \unit[$10^{-9}$]{W} \\
 $\beta$ & \unit[300]{$/\text{km}^2$}\\
 $\delta$ & 0.5\\
 $\kappa$ & \unit[20]{m} \\
 $\theta$ & \unit[0]{dB}\\
 \hline
\end{tabular}
 \label{tab:table}
\end{center}
\end{table}

In this subsection we evaluate the performance of our analytical model by generating numerical results using the results derived in the previous section as well as simulations. We simulate random distributions of user hotspots across multiple Monte Carlo (MC) trials and record the coverage probability values when UAVs are placed above the hotspot centers. In \Fig{ClusterRad} to \Fig{PPPDensity} below, solid lines denote the analytical values for the coverage probability and the markers denote results from MC trials. Unless stated otherwise the parameters used for the numerical results are from Table \ref{tab:table}. 

\begin{figure}[t!]
\centering
	\includegraphics[width=.45\textwidth]{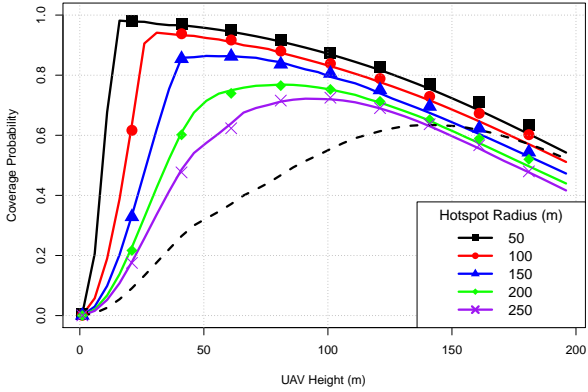}
    \vspace{-5mm}
	\caption{
	Coverage probability given a hotspot density of \unit[5]{$\text{/km}^2$} and beamwidth $\omega$ of 150 degrees. Solid lines denote analytical results, markers denote simulations, and the dashed line denotes the result for the PPP distribution of UAVs.
	}
	\label{fig:ClusterRad}
	\vspace{-3mm}
\end{figure}

\begin{figure}[t!]
\centering
	\includegraphics[width=.45\textwidth]{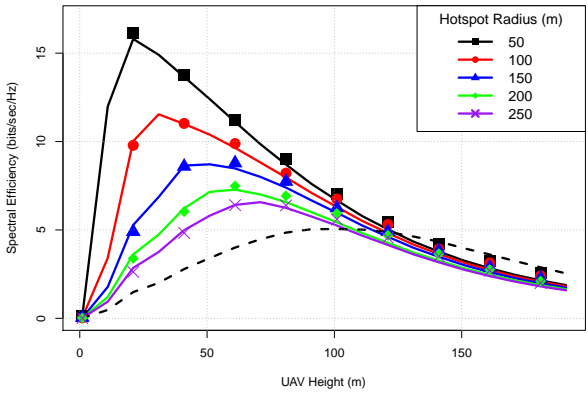}
    \vspace{-5mm}
	\caption{
	Spectral efficiency given a hotspot density of \unit[5]{$\text{/km}^2$} and beamwidth $\omega$ of 150 degrees. Solid lines denote analytical results, markers denote simulations,  and the dashed line denotes the result for the PPP distribution of UAVs.
	}
	\label{fig:ClusterRadRate}
	\vspace{-3mm}
\end{figure}

In \Fig{ClusterRad} we demonstrate the performance of the UAV network for different values of the user hotspot radius $\rmax$. We can see that when the radius is the lowest, and therefore the users are the most concentrated, the performance of the network is best. This is due to the reduced distance between a typical user and the UAV above the hotspot center, which allows the user to associate to the hotspot center UAV more often and receive a better signal from it. For each hotspot radius there is a single UAV height which maximises the coverage probability: below this height the serving UAV is more likely to have NLOS to the reference user, which decreases the received signal strength, and above this height the interfering UAVs are more likely to have LOS on the user and more interfering UAV will cast their antenna beam over the user, increasing interference. We note that the UAV height which maximises the coverage probability increases as we increase the hotspot radius; this is due to the fact that the increasing distance between a typical user and its hotspot center UAV increases the probability of a LOS-blocking building being in the way, and therefore the UAV network must increase its height to compensate. The dashed line denotes the coverage probability for the case where $\rmax \to \infty$, which is equivalent to the performance of a UAV network that ignore user hotspot locations and are positioned independently. In \Fig{ClusterRadRate} we present the spectral efficiency of the UAV network for the same parameters. We can see that the spectral efficiency curves closely match the shape of those given in \Fig{ClusterRad}, including the approximate locations of the optimum UAV height for each correspoding hotspot radius. 

It is worth noting that the range of optimum UAV heights for the smaller hotspot radii in \Fig{ClusterRad} and \Fig{ClusterRadRate} is in the order of 25-50m above ground. This height may be too low for feasible UAV network operation in an urban environment, due to factors such as wireless backhaul availability \cite{Galkin_2018} or safety regulations \cite{FAA_2016}, \cite{EASA_2017}. A possible solution is to design UAV antennas with narrower beamwidths to allow the UAVs to operate at higher altitudes, as shown in \Fig{ClusterRadBW}. We can see that decreasing the antenna beamwidth will have the effect of increasing the height the UAV network would need to operate at to maximise the coverage probability. This result matches our previously reported result in \cite{Galkin_2017} for the case of an independently distributed UAV network.

\begin{figure}[t!]
\centering
	\includegraphics[width=.45\textwidth]{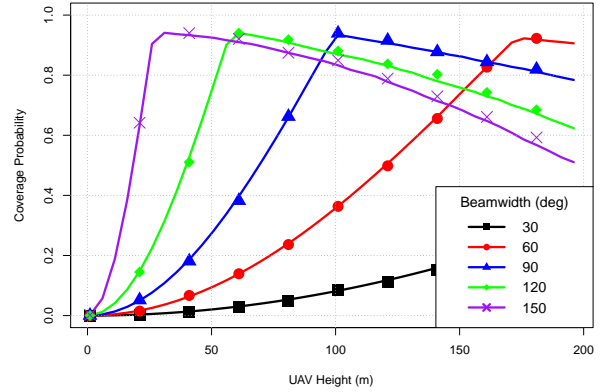}
    \vspace{-5mm}
	\caption{
	Coverage probability given a hotspot density of \unit[5]{$\text{/km}^2$} and hotspot radius of \unit[100]{m}.
	}
	\label{fig:ClusterRadBW}
	\vspace{-3mm}
\end{figure}

\begin{figure}[t!]
\centering
	\includegraphics[width=.45\textwidth]{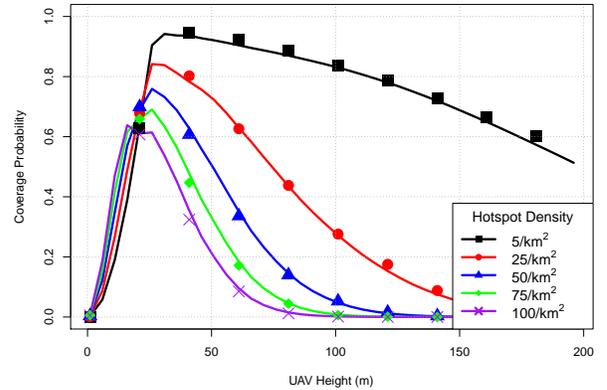}
    \vspace{-5mm}
	\caption{
	Coverage probability given a hotspot radius $\rmax$ of 100m. Solid lines denote analytical results, markers denote simulations.
	}
	\label{fig:ClusterDensity}
	\vspace{-3mm}
\end{figure}

In \Fig{ClusterDensity} we consider the network performance when the density of the user hotspots (and therefore the UAV network) is increased. We can clearly see that for greater hotspot densities the coverage probability deteriorates, due to the greater number of UAV interferers, which is not offset by the greater number of candidate serving UAVs for the user. It is also worth noting that the optimum UAV height appears to change very little for the different densities; following the results of the previous plots we conclude that the optimum height of a UAV network is primarily determined by the radius of the user hotspots, rather than the number of hotspots in a given area. This result is very different from the independent UAV placement case. \Fig{PPPDensity} presents the coverage probability for the case when the UAV network is placed independently of the hotspots; we can see that varying the UAV density will vary the optimum height, but the maximum achieveable coverage probability is approximately the same, irrespective of UAV density.

\begin{figure}[t!]
\centering
	\includegraphics[width=.45\textwidth]{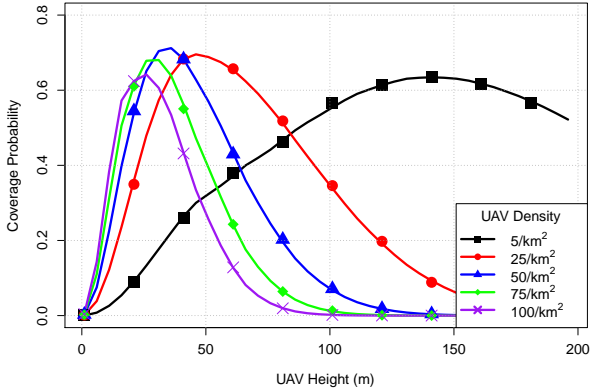}
    \vspace{-5mm}
	\caption{
	Coverage probability for the case when the UAVs are distributed indendently of the user hotspots. Solid lines denote analytical results, markers denote simulations.
	}
	\label{fig:PPPDensity}
	\vspace{-3mm}
\end{figure}

\subsection{UAV Placement Comparison}

\begin{figure}[t!]
\centering
	\includegraphics[width=.45\textwidth]{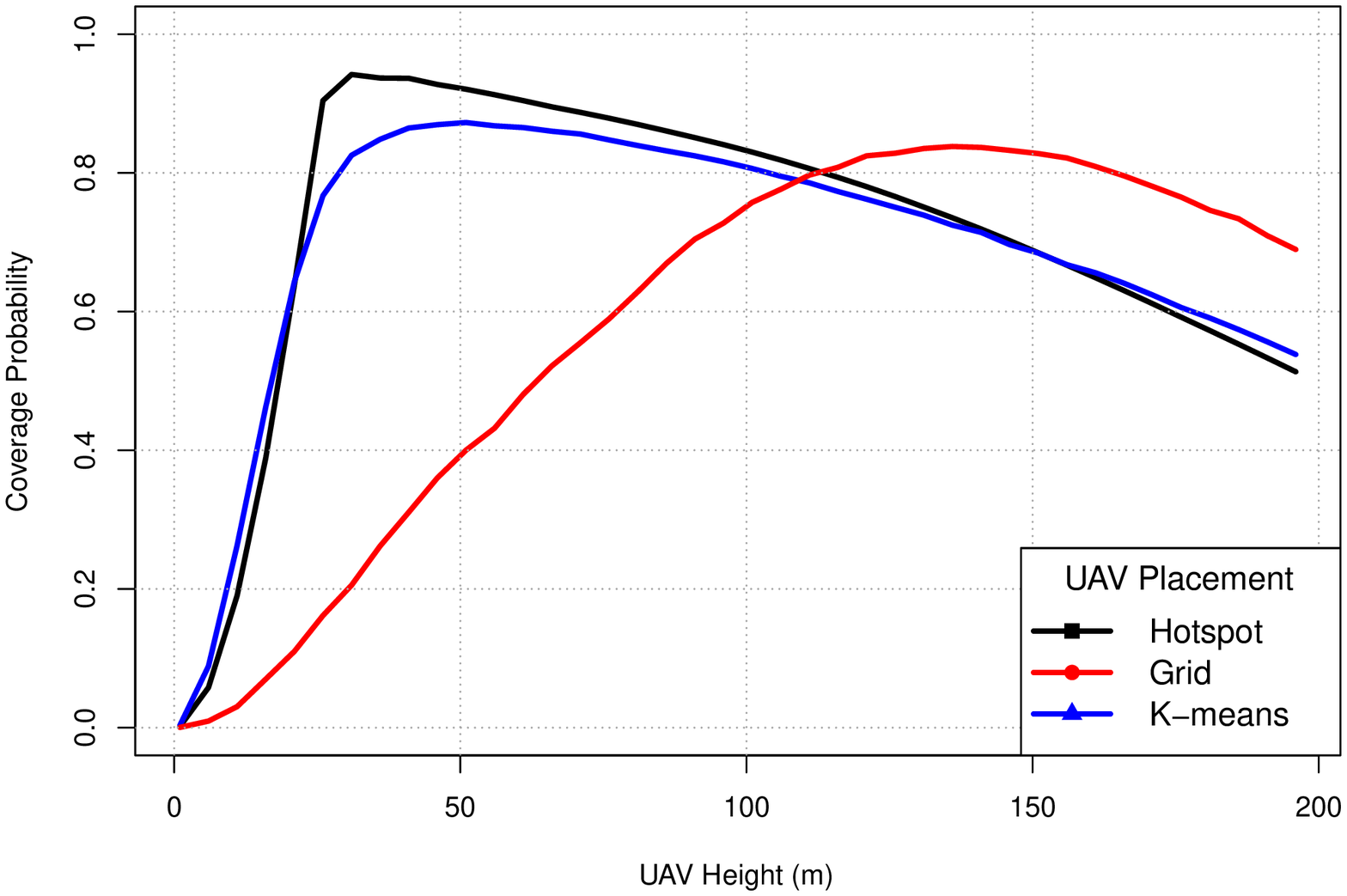}
    \vspace{-5mm}
	\caption{
	Coverage probability for different UAV placement strategies, given a hotspot radius $\rmax$ of 100m and a UAV density of $5/km^2$.
	}
	\label{fig:Placement5}
	\vspace{-3mm}
\end{figure}

\begin{figure}[t!]
\centering
	\includegraphics[width=.45\textwidth]{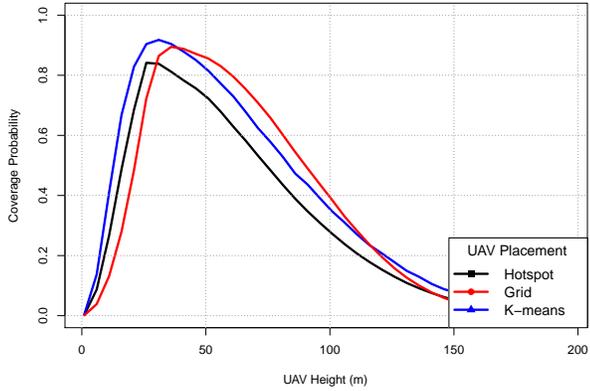}
    \vspace{-5mm}
	\caption{
		Coverage probability for different UAV placement strategies, given a hotspot radius $\rmax$ of 100m and a UAV density of $25/km^2$.
	}
	\label{fig:Placement25}
	\vspace{-3mm}
\end{figure}

In \Fig{Placement5} and \Fig{Placement25} we compare the performance of the UAV network when the UAVs are positioned at user hotspot centers, according to a rectangular grid which is independent of user locations and following K-means clustering placement. K-means clustering is a heuristic algorithm which partitions a given set of points in space into cells of approximately equal size and finds the centroid of the cells. In our previous work \cite{Galkin_2016} we have demonstrated how this algorithm can be used to optimally position UAVs around known user locations. For the low UAV density scenario in \Fig{Placement5} we can see that putting UAVs at hotspot centers gives the best overall performance, due to the shorter distance between users and their serving UAVs. The behaviour changes when we increase the UAV density in \Fig{Placement25}. Given a larger number of UAVs (and therefore interferers) the network benefits from putting UAVs in a grid pattern, even though this pattern does not take into account the actual locations of the users. At higher UAV densities the network benefits from a UAV placement strategy which maximises the distances between UAVs and which therefore limits the UAV coverage overlap and interference, even if this placement strategy does not necessarily reduce the distances between users and their serving UAV. The K-means algorithm is able to give the best performance as it both minimises the distance between the users and their serving UAV while also spreading the UAVs out across cells of roughly equal size, which separates out interferers. 

\begin{figure}[t!]
\centering
	\includegraphics[width=.45\textwidth]{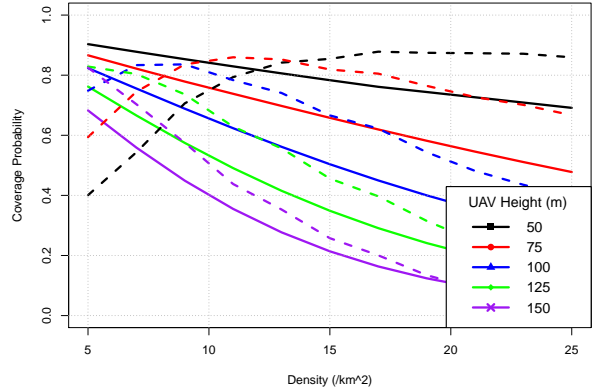}
    \vspace{-5mm}
	\caption{
	Coverage probability given a hotspot radius $\rmax$ of 100m. Solid lines denote results for UAVs placed at user hotspots, dashed lines denote UAVs arranged in a rectangular grid.
	}
	\label{fig:GridDensity}
	\vspace{-3mm}
\end{figure}

In \Fig{GridDensity} we further explore the effect of the UAV density on the coverage probability of the two UAV placement policies. We can see that there exists a maximum UAV density for a given UAV height, above which the grid deployment will outperform the hotspot center deployment, and that this maximum density decreases as we consider larger UAV heights. It is clear that at larger UAV densities the channel is interference-limited, and therefore the best network performance is achieved only when the UAVs intelligently position themselves with respect to each other, rather than based on the locations of the end users. As the UAV density is tied to the density of demand hotspots this result also demonstrates that there exists a maximum hotspot density above which the UAV network should prioritise minimising interference rather than maximising the received signal strength.


\section{Conclusion}
In this paper we have used stochastic geometry to model a UAV network serving user hotspots in an urban environment, considering UAV network parameters such as antenna beamwidth and height above ground, as well as environmental parameters such as user hotspot radius and building density. We derived an expression for the coverage probability of the UAV network as a function of these parameters and then verified the derivation numerically, while showing the trade-offs in performance that occur under different network conditions. We then compared this network performance against the case where the UAVs are positioned according to K-means, or according to a rectangular grid. Our results showed that positioning the UAVs above user hotspots is warranted when the density of hotspots (and therefore the density of required UAVs) is sufficiently low. For higher densities, the UAVs cause interference for one another, which cancels out the benefits of intelligent placement above users. 

\section*{Acknowledgements}
This material is based upon works supported by the Science Foundation
Ireland under Grants No. 10/IN.1/I3007 and 14/US/I3110.

\section*{Appendix A}

The Laplace transform $\Lc_{I_0}(s)$ is derived as 

\begin{align}
&\Lc_{I_{0}}(s) =  \Ed\Big[\exp\Big(-s I_0\Big)\Big] \nonumber \\
&=\Ed_{R_0,T_{0}}\Big[\Ed_{H_{T_{0}}}\left[\exp\Big(-H_{T_{0}} \eta_{R_0} l(R_0,\gamma,\alpha_{T_{0}})s\Big)\right]\Big]\nonumber\\
&=\Ed_{R_0,T_{0}}\Big[g(R_0,s,m_{T_{0}},\alpha_{T_{0}})\Big]\nonumber\\
&=\sum\limits_{j\in\{\los,\nlos\}}\int\limits_{0}^{\rmax}g(r_0,s,m_{j},\alpha_{j}) f_{R_{0j}}(r_0|\bar{S}_{0}< \bar{S}_\star) \dr r_0 \nonumber\\
\label{eq:LaplaceI0}
\end{align}

\noindent
where 
\begin{equation}
g(r_0,s,m_{j},\alpha_{j}) = \left(\frac{m_{j}}{\eta_{r_0} s(r_0^2+\gamma^2)^{-\alpha_{j}/2} + m_{j}}\right)^{m_{j}}   \nonumber,
\end{equation}

\noindent
which comes from $H_{T_{0}}$ being gamma distributed, with $f_{R_{0j}}(r_0|\bar{S}_{0}<\bar{S}_\star)$ denoting the joint pdf of the hotspot center UAV's distance and channel type, given that the hotspot UAV is positioned such that its signal is weaker than the user's serving UAV signal. This is derived as

\begin{align}
& f_{R_{0j}}(r_0|\bar{S}_{0}<\bar{S}_\star) = \Pd(R_0=r_0,T_{0}=j|\bar{S}_{0}< \bar{S}_\star) \nonumber \\
& = \frac{\Pd(R_0=r_0,T_{{0}}=j,\bar{S}_{0}< \bar{S}_\star)}{\Pd(\bar{S}_{0}< \bar{S}_\star)} \nonumber \\
& \overset{(a)}{=} \frac{\Pd(\bar{S}_{0}< \bar{S}_\star|R_0=r_0,T_{{0}}=j)\Pd(T_{{0}}=j|R_0=r_0)\Pd(R_0=r_0)}{\Bc_{\texttt{0}}} \nonumber \\
& \overset{(b)}{=} \frac{\textbf{1}(c_{j}(t_\star)\leq r_0 \leq \rmax)\Pd_{j}(r_0)f_{R_{0}}(r_0)}{\Bc_{0}}
\end{align}

\noindent
 where $(a)$ follows from replacing $\Pd(S_{0}< S_\star)$ with $\Bc_{0}$ which is the probability of the hotspot center UAV providing a weaker signal than the serving UAV given in \Eq{Bc}, $(b)$ follows from $\Pd(R_0=r_0) = f_{R_0}(r_0)$, $\Pd(T_{{0}}=j|R_0=r_0) = \Pd_{j}(r_0)$ and $\Pd(S_{0}< S_\star|R_0=r_0,T_{{0}}=j) = \textbf{1}(c_{j}(t_\star)\leq r_0 \leq \rmax)$ where $\textbf{1}(.)$ is the indicator function. Inserting this expression into $\Eq{LaplaceI0}$ we can write the integral as 
 
   \begin{align}
   &\frac{2}{(\rmax^2)\Bc_{0}}\sum\limits_{j\in\{\los,\nlos\}}\int\limits_{c_j(t_\star)}^{\rmax}g(r_0,s,m_{j},\alpha_{j}) \Pd_{j}(r_0) r_0 \dr r_0 \nonumber\\
   \label{eq:Li1}
   \end{align}
   
   Recall that $\eta_{r_0}=0$ for values of $r_0 > \Rcone$, which will reduce $g(r_0,s,m_{j},\alpha_{j})$ to 1. If $\rmax > \Rcone$ then the above integral is separated into two sub-integrals as

 \begin{align}
 &\frac{2}{(\rmax^2)\Bc_{0}}\sum\limits_{j\in\{\los,\nlos\}}\Bigg(\int\limits_{c_j(t_\star)}^{\Rcone}\hspace{-2mm}g(r_0,s,m_{j},\alpha_{j}) \Pd_{j}(r_0) r_0 \dr r_0 \nonumber \\
 &+ \int\limits_{\Rcone}^{\rmax}\Pd_{j}(r_0)r_0 \dr r_0\Bigg)
 \label{eq:Li2}
 \end{align}
 
 The integrals in \Eq{Li1} and \Eq{Li2} assume that $c_{j}(t_\star)<\min(\rmax,\Rcone)$. For the case where $\Rcone\leq c_{j}(t_\star)\leq \rmax$ the first sub-integral in \Eq{Li2} reduces to 0 and the second sub-integral takes $c_{j}(t_\star)$ as its lower integration bound, if $c_{j}(t_\star)\geq \max(\Rcone,\rmax)$ then both integrals reduce to 0. From the definition \Eq{LOS} the LOS probability is a step function, therefore the integral $\int\limits_{c_j(t_\star)}^{\Rcone}\hspace{-2mm}g(r_0,s,m_{j},\alpha_{j}) \Pd_{j}(r_0) r_0 \dr r_0$ can be written as a sum of weighted integrals 
 
 \begin{align}
 \sum\limits_{q=\floor*{c_j(t_\star)\sqrt{\beta\delta}}}^{\floor*{\Rcone\sqrt{\beta\delta}}} \Pd_{j}(l)\int\limits_{\Lower}^{\Upper}g(r_0,s,m_{j},\alpha_{j}) r_0 \dr r_0
\label{eq:laplaceSum}
 \end{align}
 
 \noindent
 where $l = \max(c_j(t_\star),q/\sqrt{\beta\delta})$ and $u = \min(\Rcone,(q+1)/\sqrt{\beta\delta})$. Using a derivation process similar to (11) in \cite{Galkin_2017} the integral $\int\limits_{\Lower}^{\Upper}g(r_0,s,m_{j},\alpha_{j}) r_0 \dr r_0$ can be expressed in analytical form, 

\begin{align}
&\int\limits_{\Lower}^{\Upper}\left(\frac{m_{j}}{\eta_{r_0} s(r_0^2+\gamma^2)^{-\alpha_{j}/2} + m_{j}}\right)^{m_{j}} r_0 \dr r_0 \nonumber \\
&\overset{(a)}{=} \int\limits_{(\Lower^2+\gamma^2)^{1/2}}^{(\Upper^2+\gamma^2)^{1/2}}\left(\frac{m_{j}}{\eta_{r_0}sy^{-\alpha_{j}} + m_{j}}\right)^{m_{j}} y\dr y \nonumber \\
& \overset{(b)}{=} \frac{ 1}{\alpha_{j}}\int\limits_{(\Lower^2+\gamma^2)^{\alpha_{j}/2}}^{(\Upper^2+\gamma^2)^{\alpha_{j}/2}}\left(1-\frac{1}{1 + zm_{j}(\eta_{r_0}s)^{-1}}\right)^{m_{j}}z^{2/\alpha_{j} - 1}\dr z \nonumber \\
& \overset{(c)}{=} \frac{1}{\alpha_{j}}\Bigg(\int\limits_{(\Lower^2+\gamma^2)^{\alpha_{j}/2}}^{(\Upper^2+\gamma^2)^{\alpha_{j}/2}} z^{2/\alpha_{j} - 1} \dr z \nonumber \\
 &+\sum\limits_{k=1}^{m_{j}}\binom{m_{j}}{k}(-1)^{k}\int\limits_{(\Lower^2+\gamma^2)^{\alpha_{j}/2}}^{(\Upper^2+\gamma^2)^{\alpha_{j}/2}}\frac{z^{2/\alpha_{j} - 1}}{(1+zm_{j}(\eta_{r_0}s)^{-1})^k}\dr z \Bigg) ,\nonumber \\
&\overset{(d)}{=}\frac{1}{2}\Bigg((\Upper^2-\Lower^2)+\sum\limits_{k=1}^{m_{j}}\binom{m_{j}}{k}(-1)^{k}\nonumber \\
&\cdot\bigg((\Upper^2+\gamma^2)\mbox{$_2$F$_1$}\Big(k,\frac{2}{\alpha_{j}};1+\frac{2}{\alpha_{j}};  -\frac{m_{j}(\Upper^2+\gamma^2)^{\alpha_{j}/2}}{\eta_{r_0}s}\Big) \nonumber\\
&-(\Lower^2+\gamma^2)\mbox{$_2$F$_1$}\Big(k,\frac{2}{\alpha_{j}};1+\frac{2}{\alpha_{j}};-\frac{m_{j}(\Lower^2+\gamma^2)^{\alpha_{j}/2}}{\eta_{r_0}s}\Big)\bigg)\Bigg)\Bigg) 
\label{eq:laplaceFinalvcproof}
\end{align}

\noindent
where $(a)$ stems from the substitution $y=(r_0^2+\gamma^2)^{1/2}$, $(b)$ from the substitution $z = y^{a_{j}}$, $(c)$ from applying binomial expansion and $(d)$ from using \cite{Ryzhik_2007}[Eq. 3.194.1], The solution above is inserted into \Eq{laplaceSum} which is denoted as $\Cc_j(s)$ and inserted into \Eq{laplaceFinal}.

\section*{Appendix B}

In this appendix we present an analytical expression for a higher order derivative of $\Lc_{I_0}(s)$ for the case when $I_0$ is not 0. The $i$th derivative is given as 

\begin{align}
 &\frac{d^{i}\Lc_{I_0}(s)}{ds^{i}} = \frac{1}{\rmax^2 \Bc_{\texttt{0}}} \sum_{j\in\{\los,\nlos\}} \nonumber \\
 &\cdot\Bigg(\sum\limits_{q=\floor*{c_j(t_\star)\sqrt{\beta\delta}}}^{\floor*{\min(\rmax,\Rcone)\sqrt{\beta\delta}}} \Pd_{j}(l) \sum\limits_{k=1}^{m_{j}}\binom{m_j}{k}(-1)^{k}\nonumber \\
 &\cdot\bigg(\frac{d^{i} f\left(k,j,(\Upper^2+\gamma^2),s\right)}{ds^{i}} - \frac{d^{i} f\left(k,j,(\Lower^2+\gamma^2),s\right)}{ds^{i}}\bigg)\Bigg) 
\label{eq:comp2}
\end{align} 

\noindent
where

\begin{equation}
f\left(k,j,b,s\right) =  b \mbox{$_2$F$_1$}\Big(k,\frac{2}{\alpha_j};1+\frac{2}{\alpha_j};  z(j,b,s)\Big)
\end{equation}

\noindent
and 

\begin{equation}
z(j,b,s) = -\frac{m_j b^{\alpha_j/2}}{\eta_{r_0}s} .
\label{eq:comp3}
\end{equation}

The $i$th derivative of $f\left(k,j,b,s\right)$  with respect to $s$ can be obtained using \cite{Ryzhik_2007}[0.430.1]:

\begin{align}
&\frac{d^{i} f\left(k,j,b,s\right)}{ds^{i}} = \sum\limits_{p=1}^{i} \frac{U_p}{p!} \frac{d^{p} f\left(k,j,b,s\right)}{dz^p}
\end{align}

\noindent
where 

\begin{equation}
U_p = \sum\limits_{a=0}^{p-1}(-1)^a \binom{p}{a}z(j,b,s)^{a}\frac{d^i z(j,b,s)^{p-a}}{ds^i}
\end{equation}

\begin{align}
 &\frac{d^{p} f\left(k,j,b,s\right)}{dz^p} = \nonumber \\ 
 &b \frac{k_{(p)}(2/\alpha_j)_{(p)}}{(1+2/\alpha_j)_{(p)}} \bigg(\mbox{$_2$F$_1$}\Big(k+p,\frac{2}{\alpha_j}+p;1+\frac{2}{\alpha_j}+p; z(j,b,s) \Big)\bigg)
\end{align}

\noindent
and 

\begin{align}
&\frac{d^i z(j,b,s)^{p-a}}{ds^i} = \nonumber \\
&\sum\limits_{e=0}^{i}\sum\limits_{n=0}^{e} (-1)^n(-m_j b^{\alpha_j/2}/\eta_{r_0})^e (-m_j b^{\alpha_j/2}/(\eta_{r_0}s))^{(p-a-e)} \nonumber \\ 
&\cdot\frac{s^{(-i-e)}(1+p-a-e)_{(e)}(1+n-i-e)_{(i)}}{n!(e-n)!}
\end{align}

\noindent
where $(.)_{(a)}$ is the Pochhammer notation for the rising factorial.

\ifCLASSOPTIONcaptionsoff
  \newpage
\fi



\bibliographystyle{./bib/IEEEtran}
\bibliography{./bib/IEEEabrv,./bib/IEEEfull}

\begin{thebibliography}{10}
\providecommand{\url}[1]{#1}
\csname url@samestyle\endcsname
\providecommand{\newblock}{\relax}
\providecommand{\bibinfo}[2]{#2}
\providecommand{\BIBentrySTDinterwordspacing}{\spaceskip=0pt\relax}
\providecommand{\BIBentryALTinterwordstretchfactor}{4}
\providecommand{\BIBentryALTinterwordspacing}{\spaceskip=\fontdimen2\font plus
\BIBentryALTinterwordstretchfactor\fontdimen3\font minus
  \fontdimen4\font\relax}
\providecommand{\BIBforeignlanguage}[2]{{%
\expandafter\ifx\csname l@#1\endcsname\relax
\typeout{** WARNING: IEEEtran.bst: No hyphenation pattern has been}%
\typeout{** loaded for the language `#1'. Using the pattern for}%
\typeout{** the default language instead.}%
\else
\language=\csname l@#1\endcsname
\fi
#2}}
\providecommand{\BIBdecl}{\relax}
\BIBdecl

\bibitem{Zeng_20162}
Y.~Zeng, R.~Zhang, and T.~J. Lim, ``{Wireless Communications With Unmanned
  Aerial Vehicles: Opportunities and Challenges},'' \emph{IEEE Communications
  Magazine}, vol.~54, no.~5, pp. 36--42, May 2016.

\bibitem{Heimfarth_2014}
T.~Heimfarth and J.~De~Araujo, ``{Using Unmanned Aerial Vehicle to Connect
  Disjoint Segments of Wireless Sensor Network},'' \emph{IEEE 28th
  International Conference on Advanced Information Networking and Applications
  (AINA)}, pp. 907--914, May 2014.

\bibitem{Merwaday_2015}
A.~Merwaday and I.~Guvenc, ``{UAV Assisted Heterogeneous Networks for Public
  Safety Communications},'' \emph{IEEE Wireless Communications and Networking
  Conference Workshops (WCNCW)}, pp. 329--334, March 2015.

\bibitem{BorYaliniz_20162}
I.~Bor-Yaliniz and H.~Yanikomeroglu, ``{The New Frontier in RAN Heterogeneity:
  Multi-Tier Drone-Cells},'' \emph{IEEE Communications Magazine}, vol.~54,
  no.~11, pp. 48--55, Nov. 2016.

\bibitem{FAA_2016}
\BIBentryALTinterwordspacing
{Federal Aviation Administration; Office of the Secretary of Transportation;
  Department of Transportation}, ``{Operation and Certification of Small
  Unmanned Aircraft Systems},'' \emph{14 CFR Parts 21, 43, 61, 91, 101, 107,
  119, 133, and 183}, pp. 1--624, 2016. [Online]. Available:
  \url{https://www.faa.gov/uas/media/RIN_2120-AJ60_Clean_Signed.pdf}
\BIBentrySTDinterwordspacing

\bibitem{EASA_2017}
\BIBentryALTinterwordspacing
{European Aviation Safety Agency (EASA)}, ``{NPA 2017-05 (A): Introduction of a
  Regulatory Framework for the Operation of Drones},'' 2017. [Online].
  Available:
  \url{https://www.easa.europa.eu/document-library/notices-of-proposed-amendment/npa-2017-05}
\BIBentrySTDinterwordspacing

\bibitem{SESAR_2017}
\BIBentryALTinterwordspacing
{SESAR Joint Undertaking}, ``{U-Space Blueprint},'' 2017. [Online]. Available:
  \url{http://www.sesarju.eu/u-space-blueprint}
\BIBentrySTDinterwordspacing

\bibitem{Galkin_2017}
B.~{Galkin}, J.~{Kibi{\l}da}, and L.~A. {DaSilva}, ``{Coverage Analysis for
  Low-Altitude UAV Networks in Urban Environments},'' in \emph{IEEE GLOBECOM},
  Dec. 2017, pp. 1--6.

\bibitem{Zeng_2016}
Y.~Zeng, R.~Zhang, and T.~J. Lim, ``{Throughput Maximization for UAV-Enabled
  Mobile Relaying Systems},'' \emph{IEEE Transactions on Communications},
  vol.~64, no.~12, pp. 4983--4996, Dec. 2016.

\bibitem{Zeng_2017}
Y.~Zeng and R.~Zhang, ``{Energy-Efficient UAV Communication with Trajectory
  Optimization},'' \emph{IEEE Transactions on Wireless Communications},
  vol.~16, no.~6, pp. 3747 -- 3760, March 2017.

\bibitem{Zhang_20172}
J.~Zhang, Y.~Zeng, and R.~Zhang, ``{Spectrum and Energy Efficiency Maximization
  in UAV-Enabled Mobile Relaying},'' in \emph{IEEE International Conference on
  Communications (ICC)}, May 2017, pp. 1--6.

\bibitem{BorYaliniz_2016}
R.~{Irem Bor-Yaliniz}, A.~El-Keyi, and H.~Yanikomeroglu, ``{Efficient 3-D
  Placement of an Aerial Base Station in Next Generation Cellular Networks},''
  \emph{IEEE International Conference on Communications (ICC)}, pp. 1--6, May
  2016.

\bibitem{Fotouhi_2016}
A.~Fotouhi, M.~Ding, and M.~Hassan, ``{Dynamic Base Station Repositioning to
  Improve Performance of Drone Small Cells},'' in \emph{IEEE GLOBECOM
  Workshops}, Dec 2016, pp. 1--6.

\bibitem{Chen_2017}
J.~Chen and D.~Gesbert, ``{Optimal positioning of flying relays for wireless
  networks: A LOS map approach},'' in \emph{IEEE International Conference on
  Communications (ICC)}, May 2017, pp. 1--6.

\bibitem{Ouyang_2018}
J.~Ouyang \emph{et~al.}, ``{Throughput Maximization for Laser-Powered UAV
  Wireless Communication Systems},'' \emph{ArXiv e-prints}, 2018.

\bibitem{Gruber_2016}
M.~Gruber, ``{Role of Altitude When Exploring Optimal Placement of UAV Access
  Points},'' in \emph{IEEE Wireless Communications and Networking Conference},
  Sept. 2016, pp. 1--5.

\bibitem{Galkin_2016}
B.~Galkin, J.~Kibi\l{}da, and L.~A. DaSilva, ``{Deployment of UAV-Mounted
  Access Points According to Spatial User Locations in Two-Tier Cellular
  Networks},'' in \emph{2016 Wireless Days (WD)}, March 2016, pp. 1--6.

\bibitem{Mozaffari_20162}
M.~Mozaffari \emph{et~al.}, ``{Mobile Internet of Things: Can UAVs Provide an
  Energy-Efficient Mobile Architecture?}'' in \emph{IEEE GLOBECOM}, Dec 2016,
  pp. 1--6.

\bibitem{AndrewsBaccelliGanti_2011}
J.~G. Andrews, F.~Baccelli, and R.~K. Ganti, ``{A Tractable Approach To
  Coverage and Rate in Cellular Networks},'' \emph{IEEE Transactions on
  Communications}, vol.~59, no.~11, pp. 3122--3134, Nov. 2011.

\bibitem{Saha_2017}
C.~Saha, M.~Afshang, and H.~S. Dhillon, ``{Enriched K-Tier HetNet Model to
  Enable the Analysis of User-Centric Small Cell Deployments},'' \emph{IEEE
  Transactions on Wireless Communications}, vol.~PP, no.~99, p.~1, Jan. 2017.

\bibitem{Ravi_2016}
V.~V.~C. Ravi and H.~S. Dhillon, ``{Downlink Coverage Probability in a Finite
  Network of Unmanned Aerial Vehicle (UAV) Base Stations},'' in \emph{17th IEEE
  International Workshop on Signal Processing Advances in Wireless
  Communications (SPAWC)}, July 2016, pp. 1--5.

\bibitem{Chetlur_2017}
V.~{Vardhan Chetlur} and H.~S. {Dhillon}, ``{Downlink Coverage Analysis for a
  Finite 3D Wireless Network of Unmanned Aerial Vehicles},'' \emph{IEEE
  Transactions on Communications}, vol.~65, no.~10, pp. 4543--4558, Oct. 2017.

\bibitem{Zhang_2017}
C.~Zhang and W.~Zhang, ``{Spectrum Sharing for Drone Networks},'' \emph{IEEE
  Journal on Selected Areas in Communications}, vol.~35, no.~1, pp. 136--144,
  Jan 2017.

\bibitem{Hayajneh_20162}
A.~M. Hayajneh \emph{et~al.}, ``{Drone Empowered Small Cellular Disaster
  Recovery Networks for Resilient Smart Cities},'' \emph{IEEE International
  Conference on Sensing, Communication and Networking, SECON Workshops 2016},
  June 2016.

\bibitem{MahdiAzari_20172}
M.~{Mahdi Azari} \emph{et~al.}, ``{Coexistence of Terrestrial and Aerial Users
  in Cellular Networks},'' \emph{IEEE GLOBECOM Workshops}, Dec. 2017.

\bibitem{MahdiAzari_20173}
M.~M. Azari, F.~Rosas, and S.~Pollin, ``{Reshaping Cellular Networks for the
  Sky: The Major Factors and Feasibility},'' \emph{ArXiv e-prints}, 2017.

\bibitem{Galkin_2018}
B.~{Galkin}, J.~{Kibi{\l}da}, and L.~A. {DaSilva}, ``{Backhaul For Low-Altitude
  UAVs in Urban Environments},'' in \emph{IEEE International Conference on
  Communications (ICC) (To Appear)}, May 2018, pp. 1--6.

\bibitem{Balanis_2005}
C.~A. Balanis, \emph{Antenna Theory: Analysis and Design}.\hskip 1em plus 0.5em
  minus 0.4em\relax Wiley-Interscience, 2005.

\bibitem{ITUR_2012}
``{Recommendation P.1410-5 "Propagation Data and Prediction Methods Required
  for the Design of Terrestrial Broadband Radio Access Systems Operating in a
  Frequency Range From 3 to 60 GHz"},'' ITU-R, Tech. Rep., 2012.

\bibitem{Haenggi_2013}
M.~Haenggi, \emph{{Stochastic Geometry for Wireless Networks}}.\hskip 1em plus
  0.5em minus 0.4em\relax Cambridge University Press, 2013.

\bibitem{Ryzhik_2007}
I.~Gradshteyn and I.~Ryzhik, \emph{Table of Integrals, Series, and
  Products}.\hskip 1em plus 0.5em minus 0.4em\relax Academic Press, 2007.

\end{thebibliography}

\begin{IEEEbiography}[{\includegraphics[width=1in,height=1.25in,clip]{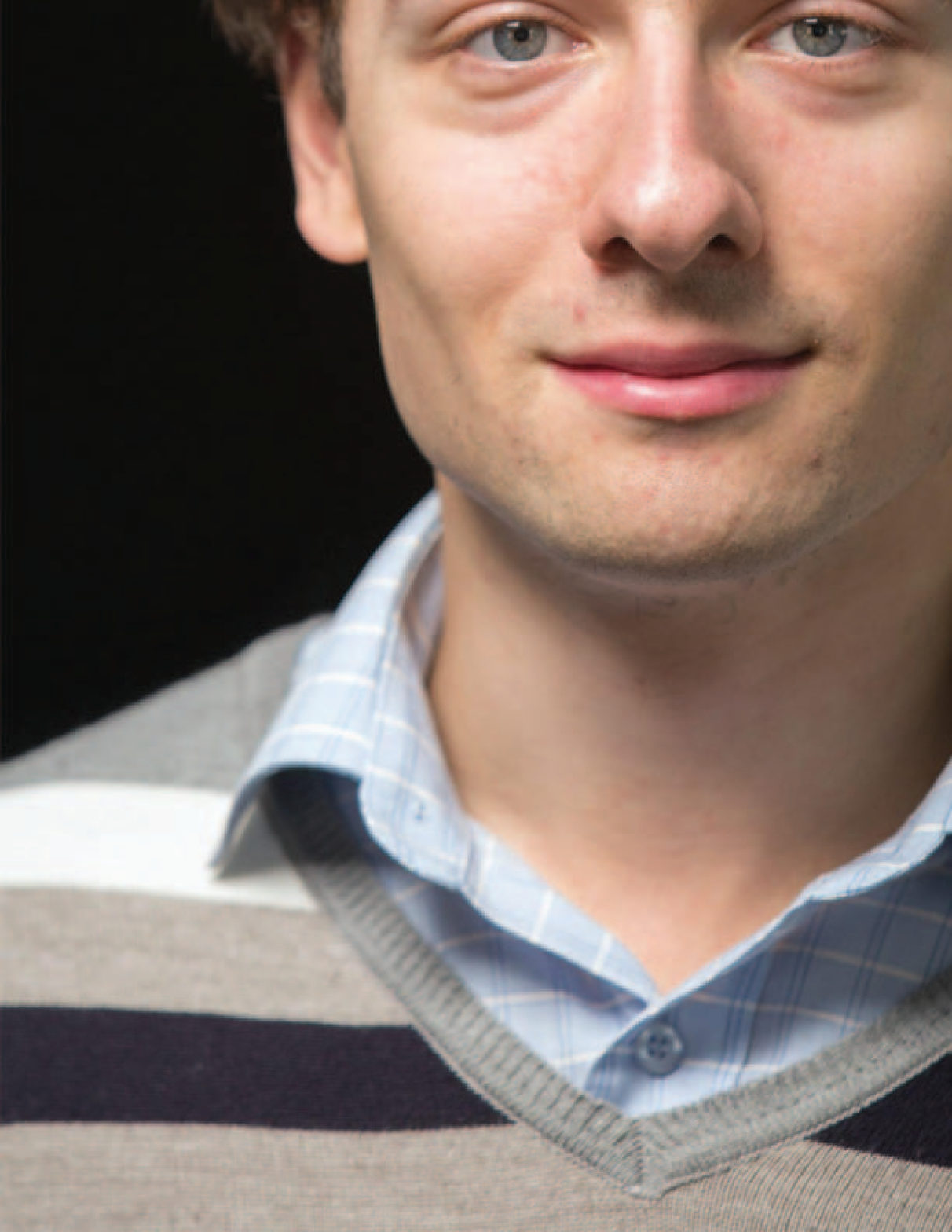}}]{Boris Galkin}
was awarded a BAI and MAI in computer \& electronic engineering by Trinity College Dublin in 2014 and since then has been working toward the Ph.D. degree at CONNECT, Trinity College, The University of Dublin, Ireland. His research interests include small cell networks, unmanned airborne devices and cognitive radios.
\end{IEEEbiography}

\begin{IEEEbiography}[{\includegraphics[width=1in,height=1.25in,clip]{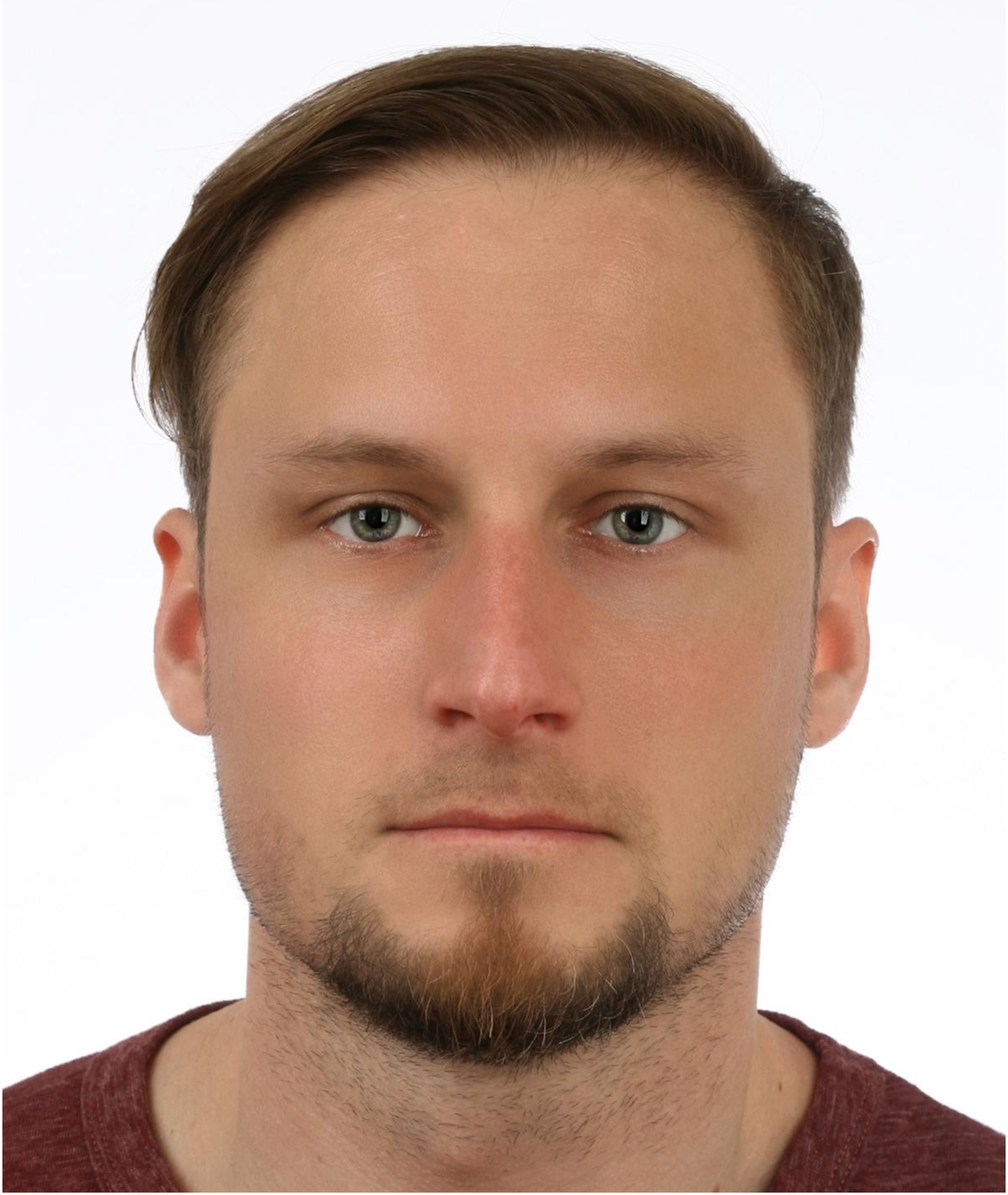}}]{Jacek Kibi\l{}da}
received the M.Sc. degree from Pozna\'n University of Technology, Poland, in 2008, and the Ph.D. degree from Trinity College, The University of Dublin, Ireland, in 2016. Currently he is a research fellow with CONNECT, Trinity College, The University of Dublin, Ireland. His research focuses on architectures and models for future mobile networks.
\end{IEEEbiography}

\begin{IEEEbiography}[{\includegraphics[width=1in,height=1.25in,clip]{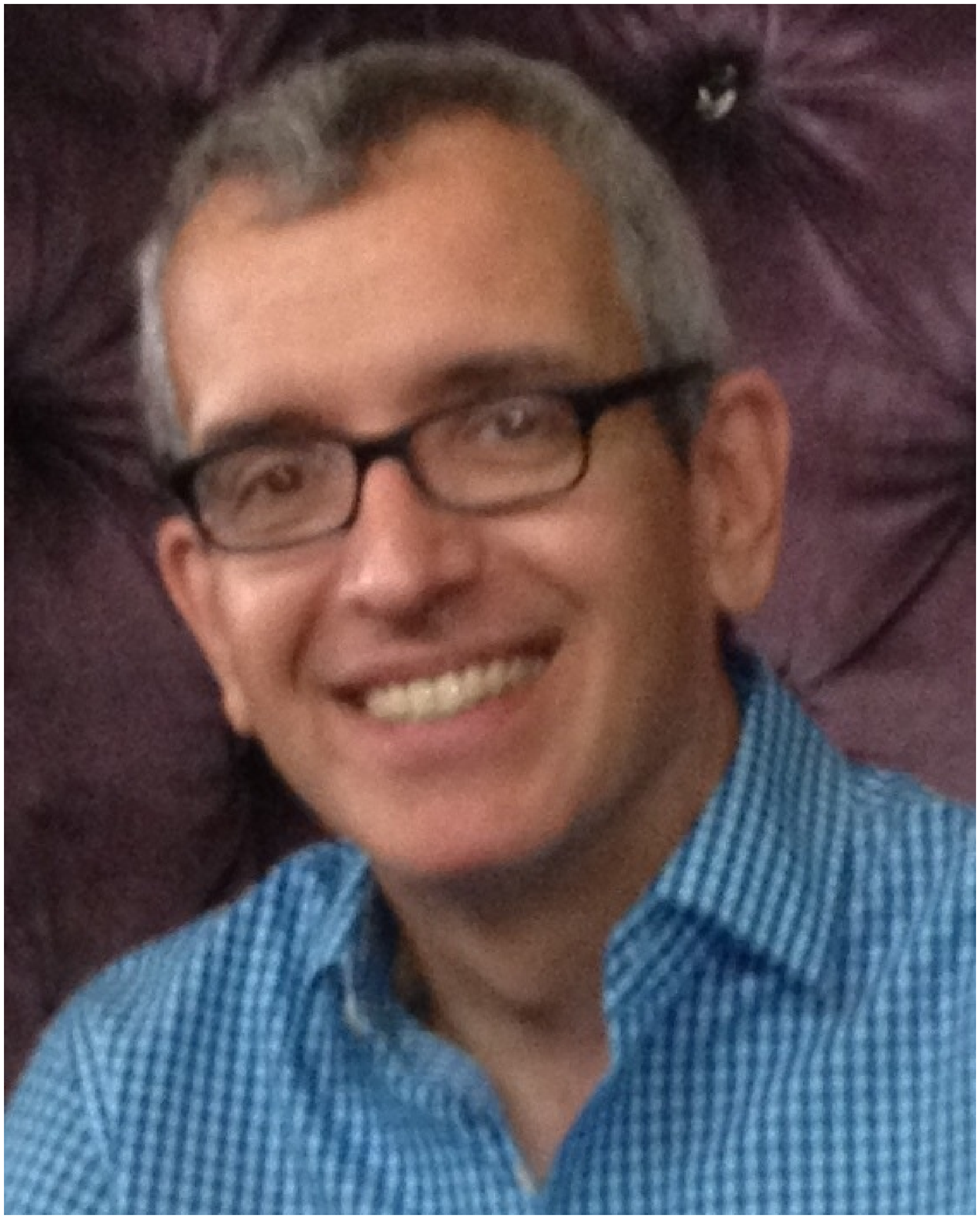}}]{Luiz A. DaSilva}
holds the chair of Telecommunications at Trinity College Dublin, where he is the director of CONNECT, telecommunications and networks research centre funded by the Science Foundation Ireland. Prior to joining Trinity College, Prof. DaSilva was with the Bradley Department of Electrical and Computer Engineering at Virginia Tech for 16 years. Prof DaSilva is a principal investigator on research projects funded by the National Science Foundation, the Science Foundation Ireland, and the European Commission. Prof DaSilva is a Fellow of Trinity College Dublin, an IEEE Communications Society Distinguished Lecturer and a Fellow of the IEEE.
\end{IEEEbiography}
\end{document}